\documentclass[aps, prx, superscriptaddress, reprint]{revtex4-2}

\usepackage{hyperref}
\usepackage{physics}
\usepackage{bbm}
\usepackage{xcolor}
\usepackage{mathtools}
\usepackage{amsmath}
\usepackage{amsfonts}
\usepackage{amsthm}
\usepackage{bm}
\usepackage{verbatim}
\usepackage{ulem}

\newtheorem{theorem}{Theorem}

\begin{document}


\title{Recovering complete positivity of non-Markovian quantum dynamics with Choi-proximity regularization}

\date{\today}

\author{Antonio D'Abbruzzo}
\email[Corresponding author: ]{antonio.dabbruzzo@sns.it}
\affiliation{Scuola Normale Superiore, 56126 Pisa, Italy}

\author{Donato Farina}
\affiliation{ICFO-Institut de Ciencies Fotoniques, The Barcelona Institute of Science and Technology, Castelldefels, 08860 Barcelona, Spain}
\affiliation{Physics Department E. Pancini - Università degli Studi di Napoli Federico II, Complesso Universitario Monte S. Angelo - Via Cintia - I-80126 Napoli, Italy.}

\author{Vittorio Giovannetti}
\affiliation{Scuola Normale Superiore, NEST, and Istituto Nanoscienze-CNR, 56126 Pisa, Italy}

\begin{abstract}
    A relevant problem in the theory of open quantum systems is the lack of complete positivity of dynamical maps obtained after weak-coupling approximations, a famous example being the Redfield master equation.
    A number of approaches exist to recover well-defined evolutions under additional Markovian assumptions, but much less is known beyond this regime.
    Here we propose a numerical method to cure the complete-positivity violation issue while preserving the non-Markovian features of an arbitrary original dynamical map.
    The idea is to replace its unphysical Choi operator with its closest physical one, mimicking recent work on quantum process tomography.
    We also show that the regularized dynamics is more accurate in terms of reproducing the exact dynamics: this allows to heuristically push the utilization of these master equations in moderate coupling regimes, where the loss of positivity can have relevant impact.
\end{abstract}

\maketitle


\section{Introduction} \label{sec:intro}

Quantum master equations constitute one of the main tools to describe quantum systems interacting with an external environment~\cite{Breuer2002}.
In their most general form they appear as matrix convolutional differential equations which are almost impossible to solve due to the complex dependence of their kernel on the microscopic details of the system~\cite{Breuer2002}.
In order to obtain more manageable equations, weak-coupling approximations are usually invoked.
However, it is well known that these manipulations can lead to a mathematical and physical disappointment: in general the resulting evolutions violate the complete positivity requirement, which basically means that they do not necessarily send quantum states into quantum states.
The most famous example of this phenomenon is given by the Redfield equation~\cite{Bloch1957,Redfield1965,Dumcke1979}, whose interplay between accuracy and unphysical predictions is still debated today~\cite{Purkayastha2016,Hartmann2020,Soret2022,Benatti2022,Tupkary2022,Tupkary2023}.

A number of approaches exist to ``regularize'' the Redfield equation into a completely positive evolution~\cite{Schaller2008,Kirsanskas2018,Farina2019,McCauley2020,Mozgunov2020,Nathan2020,Potts2021,Becker2021,Trushechkin2021,Davidovic2022,DAbbruzzo2023}, but almost all of them require additional Markovian assumptions.
This means that they should be used only in those situations where the coupling is sufficiently weak with respect to the spectral width of the environment, so that excitations injected into the environment cannot be absorbed back into the system.
Much less is known in those cases where coupling strength and spectral width are comparable in magnitude and non-Markovian effects start to arise.
In this regime we do not expect weak-coupling master equations to be extremely accurate in terms of reproducing the exact dynamics~\cite{Hartmann2020}, but, if a better approach to the problem is missing, it is still relevant to have a well-defined master equation to work with.

\begin{figure}
    \centering
    \includegraphics[scale=0.6]{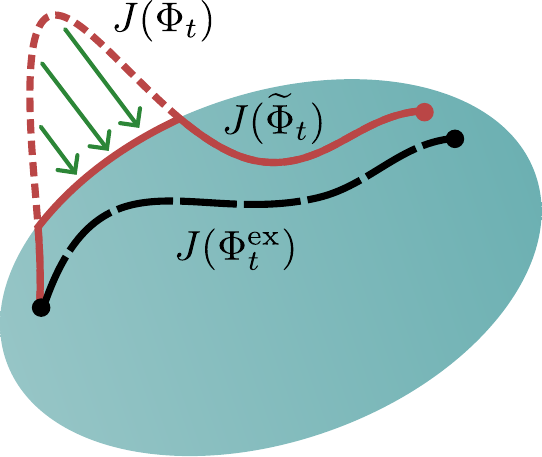}
    \caption{
        Choi-proximity regularization of a dynamical map $\Phi_t$ obtained from an approximation of the exact dynamics $\Phi_t^{\mathrm{ex}}$ (see Sec.~\ref{sec:notation} for the notations).
        In those time regions where the Choi operator $J(\Phi_t)$ exits the space of physical Choi operators (depicted as a colored oval), we project it back on such a space obtaining a CPT map $\widetilde{\Phi}_t$.
        Thanks to the convexity of the physical Choi space, the result is closer to the exact dynamics $\Phi_t^{\mathrm{ex}}$.
    }
    \label{fig:cartoon}
\end{figure}

In this paper we propose a numerical approach, called \textit{Choi-proximity regularization}, that can be applied to regularize any unphysical dynamical evolution.
The idea is to work on the corresponding Choi operator~\cite{Watrous2018} and project it onto the space of physical Choi states, which are associated with proper trace-preserving completely positive evolutions (see Fig.~\ref{fig:cartoon}).
The adoption of such a projection was inspired by recent work in quantum process tomography, where it is used to recover a proper quantum channel from a set of measurements which can possibly produce an unphysical Choi operator~\cite{Knee2018,Huang2020,Stepney2022}.
This method guarantees that the obtained dynamics does not provide unphysical predictions, while improving the accuracy (with respect to the generally unknown exact solution) in those time regions where the original map violates the requirements of physical consistency.
Moreover, no additional assumptions are made on the parameters of the model, which means that the result is still able to display eventual non-Markovian features of the original evolution.

The paper is structured as follows.
In Sec.~\ref{sec:notation} we set the stage by introducing the notation and reminding known facts about quantum channels, Choi operators, and quantum master equations for open systems.
In Sec.~\ref{sec:regularization} we discuss the Choi-proximity regularization procedure as a way to regularize any unphysical master equation.
Then, in Sec.~\ref{sec:examples}, we move on to examples of applications to an exactly-solvable qubit system and to a more complex spin-boson model.
Finally, in Sec.~\ref{sec:conclusions} we draw our conclusions.


\section{Notation and basic facts} \label{sec:notation}

Let us start with a review of basic facts about quantum channels, Choi operators, dynamical maps in open systems, and quantum master equations.
These concepts form the necessary background for what is discussed in the rest of the paper.


\subsection{Quantum channels and Choi operators}

Let us consider a quantum system $S$ represented by a complex Hilbert space $\mathcal{H}$ of finite dimension $d < \infty$.
Let $\mathrm{L}(\mathcal{H})$ be the vector space of linear operators of the form $X: \mathcal{H} \to \mathcal{H}$, endowed with the Frobenius inner product $\langle X, Y \rangle \coloneqq \Tr[X^\dagger Y]$.
Two important subsets of $\mathrm{L}(\mathcal{H})$ are the real vector space of Hermitian operators,
\begin{equation}
    \mathrm{Herm}(\mathcal{H}) \coloneqq \qty{X \in \mathrm{L}(\mathcal{H}) : X^\dagger = X},
\end{equation}
and the convex cone of positive semidefinite operators,
\begin{equation}
    \mathrm{Pos}(\mathcal{H}) \coloneqq \qty{A^\dagger A : A \in \mathrm{L}(\mathcal{H})} \subset \mathrm{Herm}(\mathcal{H}).
\end{equation}

Now, let $\mathrm{T}(\mathcal{H})$ be the vector space of linear superoperators of the form $\Phi: \mathrm{L}(\mathcal{H}) \to \mathrm{L}(\mathcal{H})$.
Such a map is called:
\textit{Hermitian-preserving} when $X \in \mathrm{Herm}(\mathcal{H})$ implies $\Phi(X) \in \mathrm{Herm}(\mathcal{H})$;
\textit{completely positive} when, for any choice of an ancillary Hilbert space $\mathcal{A}$ and for any choice of $A \in \mathrm{Pos}(\mathcal{H} \otimes \mathcal{A})$, one has $(\Phi \otimes \mathbbm{1}_{\mathrm{L}(\mathcal{H})})(A) \in \mathrm{Pos}(\mathcal{H} \otimes \mathcal{A})$, where $\mathbbm{1}_\mathcal{X}$ stands for the identity map on $\mathcal{X}$;
\textit{trace-preserving} when for any choice of $X \in \mathrm{L}(\mathcal{H})$ one has $\Tr[\Phi(X)] = \Tr[X]$.
Notice that a completely positive map is automatically Hermitian-preserving, while the converse is not true.

An important result is that there exists a bijection between $\mathrm{T}(\mathcal{H})$ and a subset of $\mathrm{L}(\mathcal{H} \otimes \mathcal{H})$.
Fix an orthonormal basis $\{e_n\}$ of $\mathcal{H}$, so that $\{E_{nm}\}$ with $E_{nm} \coloneqq e_n e_m^\dagger$ defines an orthonormal basis of $\mathrm{L}(\mathcal{H})$.
Then the map $J: \mathrm{T}(\mathcal{H}) \to \mathrm{L}(\mathcal{H} \otimes \mathcal{H})$ defined by
\begin{equation} \label{eq:def-choi}
    J(\Phi) \coloneqq \frac{1}{d} \sum_{n,m=1}^d \Phi(E_{nm}) \otimes E_{nm}
\end{equation}
is a bijection, which is inverted by
\begin{equation}\label{eq:inverse_choi}
    \Phi(X) = d \Tr_2[J(\Phi)(\mathbbm{1}_{\mathcal{H}} \otimes X^T)],
    \quad
    X \in \mathrm{L}(\mathcal{H}),
\end{equation}
where $\Tr_2 \equiv \mathbbm{1}_{\mathrm{L}(\mathcal{H})} \otimes \Tr$ is the partial trace over the second factor of the tensor product.
The map $J$ is known as \textit{Choi-Jamio{\l}kowski isomorphism}, and $J(\Phi)$ is called \textit{Choi operator} associated with $\Phi$.
It can be used to characterize the sets of superoperators defined above~\cite{Watrous2018}.
\begin{theorem}\label{thm:characterization_choi}
    $\Phi \in \mathrm{T}(\mathcal{H})$ is Hermitian-preserving if and only if $J(\Phi) \in \mathrm{Herm}(\mathcal{H} \otimes \mathcal{H})$, and it is completely positive if and only if $J(\Phi) \in \mathrm{Pos}(\mathcal{H} \otimes \mathcal{H})$.
    Moreover, $\Phi$ is trace-preserving if and only if
    \begin{equation}
        \Tr_1 J(\Phi) = \frac{\mathbbm{1}_\mathcal{H}}{d},
    \end{equation}
    where $\Tr_1 \equiv \Tr \otimes \mathbbm{1}_{\mathrm{L}(\mathcal{H})}$ is the partial trace over the first factor of the tensor product.
\end{theorem}

Particularly important in quantum physics are the \textit{CPT maps} or \textit{quantum channels}, defined as those superoperators that are both completely positive and trace-preserving.
It can be argued that a physical evolution that sends quantum states into quantum states must be represented by a quantum channel~\cite{Breuer2002}.
We introduce the symbol $\mathfrak{J}(\mathcal{H})$ for the set of \textit{physical Choi operators}, which are those operators that are Choi operators associated with some quantum channel acting on $\mathcal{H}$.
According to Theorem~\ref{thm:characterization_choi},
\begin{equation}
    \mathfrak{J}(\mathcal{H}) = \qty{P \in \mathrm{Pos}(\mathcal{H} \otimes \mathcal{H}) : \Tr_1 P = \frac{\mathbbm{1}_\mathcal{H}}{d}},
\end{equation}
and it can be proved that $\mathfrak{J}(\mathcal{H})$ is a closed and convex subset of $\mathrm{L}(\mathcal{H} \otimes \mathcal{H})$, being the intersection of the positive semidefinite cone with an affine subspace~\cite{Watrous2018}.


\subsection{Open systems and dynamical maps}

Suppose our system $S$ is coupled to an environment $E$, so that the Hamiltonian of the total system $S+E$ can be written as
\begin{equation}
    H = H_S \otimes \mathbbm{1}_E + \mathbbm{1}_S \otimes H_E + H_I,
\end{equation}
where $H_0 \equiv H_S \otimes \mathbbm{1}_E + \mathbbm{1}_S \otimes H_E$ contains free terms and $H_I$ is an interaction term between $S$ and $E$.
Let $\rho(t)$ be the density operator of $S$ at time $t$, and assume the initial density operator at time $t=0$ of the total system $S+E$ to be factorized as $\rho(0) \otimes \rho_E$, where $\rho_E$ is the initial state of the environment $E$.
Then the evolution of $S$ can be formally described by the following family $\{\Phi_t\}_{t \geq 0}$ of quantum channels:
\begin{equation}\label{eq:dynamical-maps}
    \rho(t) = \Phi_t(\rho(0)) = \Tr_E[U_t (\rho(0) \otimes \rho_E) U_t^\dagger],
\end{equation}
where $\Tr_E$ is the partial trace over $E$ and $U_t$ is the unitary time evolution operator of $S+E$.

An important property that dynamical maps $\Phi_t$ might have that we will need later is \textit{completely positive divisibility}, often abbreviated as \textit{CP-divisibility}~\cite{Wolf2008}.
A family $\{\Phi_t\}_{t \geq 0}$ of quantum channels is called CP-divisible when, for any $t \geq s$, there exist a quantum channel $\Lambda_{t,s}$ such that $\Phi_t = \Lambda_{t,s} \circ \Phi_s$.
This property is often taken as synonym of Markovianity (even though other points of view exist~\cite{Milz2019}).
If $\{\Phi_t\}$ is CP-divisible one has the Breuer-Laine-Piilo (BLP) condition~\cite{Breuer2009}, which says that given two density operators $\rho$ and $\sigma$,
\begin{equation}\label{eq:BLP}
    \frac{d}{dt} D_t(\rho, \sigma) \leq 0,
    \quad
    \forall t \geq 0,
\end{equation}
where $D_t(\rho, \sigma)$ is the Helstrom distinguishability measure between the states $\rho$ and $\sigma$~\cite{Watrous2018} after time $t$ has passed:
\begin{equation}\label{eq:distinguishability}
    D_t(\rho, \sigma) \coloneqq \frac{1}{2} \norm{\Phi_t(\rho) - \Phi_t(\sigma)}_1,
\end{equation}
with $\norm{X}_1 \coloneqq \Tr \sqrt{X^\dagger X}$ being the trace norm.
Violations of the BLP condition provide a clear signature of non-Markovianity.
Note that if $\Phi_t$ is CPT we must have that $\Phi_t(\rho)$ and $\Phi_t(\sigma)$ are both density matrices: as a consequence, $0 \leq D_t(\rho, \sigma) \leq 1$.


\subsection{Quantum master equations} \label{subsec:master-equations}

It is common to cast Eq.~\eqref{eq:dynamical-maps} as a convolutional master equation of the Nakajima-Zwanzig form
\begin{equation}\label{eq:nakajima}
    \partial_t \Phi_t = \int_0^t ds \, \mathcal{K}_{t-s} \circ \Phi_s,
\end{equation}
or, in some cases, in the time-local form
\begin{equation}\label{eq:time-local}
    \partial_t \Phi_t = \mathcal{L}_t \circ \Phi_t.
\end{equation}
In both cases, $\mathcal{K}_{t-s}$ and $\mathcal{L}_t$ consist in highly non-trivial expressions in terms of microscopic properties of the system, and approximations are needed in order to get results in practical scenarios~\cite{Breuer2002}.
For example, it is common to perform perturbative approximations in the ``strength'' of $H_I$.
To do that, it is convenient to introduce the interaction-picture operators $H_I(t) \coloneqq e^{iH_0 t} H_I e^{-iH_0 t}$ and $\varrho(t) \coloneqq e^{iH_S t} \rho(t) e^{-iH_S t}$, so that, under the standard assumption $\Tr_E[H_I(t) \rho_E] = 0$, the lowest-order expressions of Eq.~\eqref{eq:nakajima} and \eqref{eq:time-local} are given, respectively, by the \textit{Born equation}
\begin{equation}\label{eq:born}
    \frac{d}{dt} \varrho(t) = -\int_0^t ds \Tr_E[H_I(t), [H_I(s), \varrho(s) \otimes \rho_E]]
\end{equation}
and the \textit{Redfield equation}~\cite{Bloch1957,Redfield1965}
\begin{equation} \label{eq:redfield}
    \frac{d}{dt} \varrho(t) = -\int_0^t ds \Tr_E[H_I(t), [H_I(s), \varrho(t) \otimes \rho_E]].
\end{equation}
Notice how the only difference is that in the Redfield equation $\varrho(t)$ appears in the right-hand side instead of $\varrho(s)$.
The term ``Redfield equation'' is also used in the literature to indicate the variant of Eq.~\eqref{eq:redfield} in which $t$ is replaced by $\infty$ in the upper integration limit.
Here we use the terms \textit{time-dependent} Redfield equation to indicate Eq.~\eqref{eq:redfield} and \textit{time-independent} Redfield equation to indicate its variant with $\infty$ as the upper integration limit.

It is easy to verify that the dynamical maps $\Phi_t$ corresponding to Eqs.~\eqref{eq:born} and \eqref{eq:redfield} are Hermitian-preserving and trace-preserving.
However, it can be shown that they are not completely positive~\cite{Breuer2002,Dumcke1979}.

In the following it will be useful to remind the standard Kossakowski form of the Redfield equation~\cite{DAbbruzzo2023}.
Let $\{\ket{k}\}$ be the orthonormal basis of normalized eigenvectors of $H_S$, so that $H_S \ket{k} = \omega_k \ket{k}$, and call $E_{kq} \coloneqq \dyad{k}{q}$.
Moreover, suppose the interaction is of the form $H_I = \sum_\alpha L_{\alpha} \otimes B_{\alpha}$.
It is convenient to express the result in Schr\"odinger picture, where exponential factors of the form $e^{\pm i \omega_k t}$ disappear~\cite{Breuer2002}.
Then,
\begin{multline} \label{eq:redfield-standard}
    \frac{d}{dt} \rho(t) = -i[H_S + H_{LS}(t), \rho(t)] \\
    + \sum_{k,q,n,m} \chi_{kq,nm}(t) \qty[E_{kq} \rho(t) E_{nm}^\dagger - \frac{1}{2} \{E_{nm}^\dagger E_{kq}, \rho(t)\}],
\end{multline}
where $H_{LS}(t) = \sum_{k,q,n,m} \eta_{kq,nm}(t) E_{nm}^\dagger E_{kq}$ and
\begin{gather}
    \chi_{kq,nm}(t) = \sum_{\alpha,\beta} \qty[F_{\alpha\beta}(\omega_{kq},t) + F^*_{\beta\alpha}(\omega_{nm},t)] L_{\beta,kq} L^*_{\alpha,nm}, \\
    \eta_{kq,nm}(t) = \sum_{\alpha,\beta} \frac{F_{\alpha\beta}(\omega_{kq},t) - F^*_{\beta\alpha}(\omega_{nm},t)}{2i} L_{\beta,kq} L^*_{\alpha,nm}.
\end{gather}
Here $L_{\alpha,kq} \coloneqq \mel{k}{L_\alpha}{q}$ is the matrix element of the system operator $L_\alpha$, $\omega_{kq} \coloneqq \omega_q - \omega_k$ is a Bohr transition frequency, and
\begin{equation} \label{eq:F}
    F_{\alpha\beta}(\omega, t) = \int_0^t d\tau \, c_{\alpha\beta}(\tau) e^{i\omega\tau}
\end{equation}
with $c_{\alpha\beta}(\tau) \coloneqq \Tr[e^{iH_E \tau} B_\alpha^\dagger e^{-iH_E \tau} B_\beta \rho_E]$ environment correlation function.
The matrix $\chi(t)$ is sometimes called Kossakowski matrix associated with Eq.~\eqref{eq:redfield-standard}, and it characterizes the non-unitary part of the evolution.

In the time-independent case, the well-known Lindblad theorem~\cite{Lindblad1976,Gorini1976} tells us that positivity of $\chi$ is equivalent to the complete positivity of the dynamics.
In the more difficult time-dependent scenario, the positivity condition $\chi(t) \geq 0$ is only sufficient to guarantee complete positivity of Eq.~\eqref{eq:redfield-standard} \textit{and} CP-divisibility.
Based on this observation, the following approach to regularize the Redfield equation was proposed in Ref.~\cite{DAbbruzzo2023}: substitute $\chi(t)$ with its closest positive semidefinite matrix of the same dimension at every time point $t$.

In App.~\ref{app:vectorization} we show how to transform Eq.~\eqref{eq:redfield-standard} into a vector differential equation, which can then be solved with standard routines.
The transformation is carried out by keeping at every stage explicit access to the Kossakowski matrix $\chi(t)$, in order to easily allow the simultaneous study of the Markovian regularization in Ref.~\cite{DAbbruzzo2023}.
To our knowledge, this is not immediately possible with existing Redfield solvers, as the one contained in the Python library QuTiP~\cite{Johansson2013}.


\section{Choi-proximity regularization} \label{sec:regularization}

\subsection{Projection onto the Choi space} \label{subsec:projection}

Let $\Phi \in \mathrm{T}(\mathcal{H})$ be a Hermitian-preserving map, with associated Choi operator $P \coloneqq J(\Phi)$.
By Theorem~\ref{thm:characterization_choi} we know that $P$ is Hermitian, but in general it is not a physical Choi operator.
However, we can ask what is the physical Choi operator $\widetilde{P}$ that is ``closest'' to $P$ in the following sense:
\begin{equation}\label{eq:projection-choi}
    \widetilde{P} \coloneqq \operatorname*{arg\,min}_{X \in \mathfrak{J}(\mathcal{H})} \norm{P-X},
\end{equation}
where $\norm{X}^2 \coloneqq \langle X,X \rangle$ is the Frobenius norm.
Since $\mathfrak{J}(\mathcal{H})$ is a nonempty closed and convex set, a standard result in convex analysis guarantees that the ``projection'' $P \mapsto \widetilde{P}$ is well defined~\cite{Urruty1993}.
We argue that the map $\widetilde{\Phi}$ obtained from $\widetilde{P}$ using Eq.~\eqref{eq:inverse_choi} could work as a minimally-invasive CPT substitute for the original map $\Phi$.
Obviously, if $\Phi$ is already a quantum channel then $\widetilde{P} = P$ and $\widetilde{\Phi} = \Phi$, hence no modification is introduced.
The choice of the Frobenius norm over any other norm is mostly of mathematical convenience.
However, it can be given physical significance in terms of maximization of the likelihood function in measurement experiments~\cite{Smolin2012}.

In general, it is not possible to obtain a closed form expression for~\eqref{eq:projection-choi}, and one has to rely on numerical approaches.
In particular, \eqref{eq:projection-choi} belongs to the family of \textit{semidefinite least-squares problems}, which form a particularly interesting subclass of nonlinear convex programming with a variety of applications in numerical linear algebra and statistics~\cite{Boyd2004,Malick2004,Henrion2011}.
The interested reader can find in App.~\ref{app:algorithm} a brief survey of the most popular approaches to solve~\eqref{eq:projection-choi}, together with a more detailed description of the method chosen by us here to generate the results in Sec.~\ref{sec:examples}.
In the context of quantum process tomography, these methods have successfully been applied for Hilbert spaces with at most seven qubits ($d = 2^7$), obtaining results with high precision~\cite{Stepney2022}.


\subsection{Application to dynamical maps}

Suppose to have a set $\{\Phi_t\}_{t \geq 0}$ of Hermitian-preserving dynamical maps, and consider the new set $\{\widetilde{\Phi}_t\}_{t \geq 0}$ of quantum channels obtained by applying the procedure described above for every time point $t$.
The map $\widetilde{\Phi}_t$ can then be seen as a minimally-invasive regularized CPT version of $\Phi_t$.

To see how the maps $\Phi_t$ and $\widetilde{\Phi}_t$ differs, consider their difference in the Choi representation:
\begin{equation} \label{eq:cpt-violation}
    \Delta(t) \coloneqq J(\widetilde{\Phi}_t) - J(\Phi_t).
\end{equation}
The norm $\norm{\Delta(t)}$ can be taken as a measure of ``CPT violation'' of $\Phi_t$, and hence as a quantification of the influence of the regularization procedure.

We can write
\begin{align}
    \widetilde{\Phi}_t(\rho) &= d \Tr_2[\widetilde{P}(t) (\mathbbm{1}_{\mathcal{H}} \otimes \rho^T)] \nonumber \\
    &= \Phi_t(\rho) + d \Tr_2[\Delta(t) (\mathbbm{1}_{\mathcal{H}} \otimes \rho^T)].
\end{align}
Assuming $\Delta(t)$ to be differentiable
\footnote{The differentiability of $\Delta(t)$ can be assessed using the implicit function theorem on the Karush-Kuhn-Tucker conditions associated with the optimization problem~\eqref{eq:projection-choi}, as shown, e.g., in Ref.~\cite{Barratt2019}.
Note that even if $\Delta(t)$ happens to have regions of non-differentiability the regularization procedure can still be used from an exclusively numerical standpoint without making reference to formal master equations.}
we can write formal master equations for the regularized dynamics.
Specifically, if $\Phi_t$ satisfies the Nakajima-Zwanzig equation~\ref{eq:nakajima} we have
\begin{equation}
    \partial_t \widetilde{\Phi}_t(\rho) = \int_0^t ds \, \mathcal{K}_{t-s} \Phi_s(\rho) + d \Tr_2[\partial_t \Delta(t) (\mathbbm{1}_{\mathcal{H}} \otimes \rho^T)].
\end{equation}
This can be rewritten as a master equation for the regularized dynamics $\widetilde{\Phi}_t$ only as
\begin{align} \label{eq:reg-time-nonlocal}
    \partial_t \widetilde{\Phi}_t(\rho) &= \int_0^t ds \, \mathcal{K}_{t-s} \qty[ \widetilde{\Phi}_s(\rho) - d \Tr_2[\Delta(s) (\mathbbm{1}_{\mathcal{H}} \otimes \rho^T)] ] \nonumber \\
    &+ d \Tr_2[\partial_t \Delta(t) (\mathbbm{1}_{\mathcal{H}} \otimes \rho^T)].
\end{align}
With similar steps, we can write a similar equation in case $\Phi_t$ satisfies the time-local master equation in Eq.~\eqref{eq:time-local}:
\begin{align} \label{eq:reg-time-local}
    \partial_t \widetilde{\Phi}_t(\rho) &= \mathcal{L}_t \qty[ \widetilde{\Phi}_t(\rho) - d \Tr_2[\Delta(t) (\mathbbm{1}_{\mathcal{H}} \otimes \rho^T)] ] \nonumber \\
    &+ d \Tr_2[\partial_t \Delta(t) (\mathbbm{1}_{\mathcal{H}} \otimes \rho^T)].
\end{align}
These are the formal master equations obeyed by the regularized dynamics.
The most important fact we can get from these equations is that they contain a dependence on the initial state $\rho$ which was absent in the original evolution.
Even in case the original generator does not depend on time, $\mathcal{L}_t \equiv \mathcal{L}$, the presence of $\Delta(t)$ in the right-hand side of Eq.~\eqref{eq:reg-time-local} makes the regularized dynamics, in general, non-Markovian.
Note, however, that the ``amount of non-Markovianity'' may still be different from the original dynamics and the appearance of $\rho$ in Eqs.~\eqref{eq:reg-time-nonlocal}-\eqref{eq:reg-time-local} is taken here just as a hint that non-Markovianity can be observed, a fact that we will verify in Sec.~\ref{sec:examples} through examples.

Another important point is how the two maps $\Phi_t$ and $\widetilde{\Phi}_t$ relate to the exact solution $\Phi_t^{\mathrm{ex}}$ of the open system setting.
Since the regularization is constructed as a projection onto the physical Choi space, and since the Choi operator of the exact solution is obviously physical, the regularized dynamics $\widetilde{\Phi}_t$ is necessarily ``closer'' to $\Phi_t^{\mathrm{ex}}$ than $\Phi_t$.
This is a consequence of the convexity of $\mathfrak{J}(\mathcal{H})$, and it follows from general results in convex analysis~\cite{Urruty1993}.
Given its importance for our discussion, we reproduce here a proof.
\begin{theorem} \label{thm:distance-from-exact}
    $\norm{J(\Phi_t) - J(\Phi_t^{\mathrm{ex}})} \geq \norm*{J(\widetilde{\Phi}_t) - J(\Phi_t^{\mathrm{ex}})}$.
\end{theorem}
\begin{proof}
    First note that since $J(\widetilde{\Phi}_t)$ and $J(\Phi_t^{\mathrm{ex}})$ are in the physical Choi space, also $\alpha J(\Phi_t^{\mathrm{ex}}) + (1-\alpha) J(\widetilde{\Phi}_t)$ is in the physical Choi space for any $\alpha \in [0,1]$.
    Such an operator must be necessarily more distant to $J(\Phi_t)$ than $J(\widetilde{\Phi}_t)$, since $J(\widetilde{\Phi}_t)$ is defined as a projection.
    Hence
    \begin{widetext}
        \begin{multline}
            \frac{1}{2} \norm*{J(\Phi_t) - J(\widetilde{\Phi}_t)}^2
            \leq \frac{1}{2} \norm*{J(\Phi_t) - [\alpha J(\Phi_t^{\mathrm{ex}}) + (1-\alpha)J(\widetilde{\Phi}_t)]}^2
            = \frac{1}{2} \norm*{[J(\Phi_t) - J(\widetilde{\Phi}_t)] + \alpha [J(\widetilde{\Phi}_t) - J(\Phi_t^{\mathrm{ex}})]}^2 \\
            = \frac{1}{2} \norm*{J(\Phi_t) - J(\widetilde{\Phi}_t)}^2 + \frac{\alpha^2}{2} \norm*{J(\widetilde{\Phi}_t) - J(\Phi_t^{\mathrm{ex}})}^2
            + \alpha \langle J(\Phi_t) - J(\widetilde{\Phi}_t), J(\widetilde{\Phi}_t) - J(\Phi_t^{\mathrm{ex}}) \rangle.
        \end{multline}
        This inequality can be rewritten as
        \begin{equation}
            \langle J(\Phi_t) - J(\widetilde{\Phi}_t), J(\widetilde{\Phi}_t) - J(\Phi_t^{\mathrm{ex}}) \rangle + \frac{\alpha}{2} \norm*{J(\widetilde{\Phi}_t) - J(\Phi_t^{\mathrm{ex}})}^2 \geq 0.
        \end{equation}
        Sending $\alpha \to 0$ we conclude that
        \begin{equation}
            \langle J(\Phi_t) - J(\widetilde{\Phi}_t), J(\widetilde{\Phi}_t) - J(\Phi_t^{\mathrm{ex}}) \rangle \geq 0.
        \end{equation}
        This can be used to write
        \begin{multline}
            \norm{J(\Phi_t) - J(\Phi_t^{\mathrm{ex}})}^2 = \norm*{J(\Phi_t) - J(\widetilde{\Phi}_t) + J(\widetilde{\Phi}_t) - J(\Phi_t^{\mathrm{ex}})}^2 \\
            = \norm*{J(\Phi_t) - J(\widetilde{\Phi}_t)}^2 + \norm*{J(\widetilde{\Phi}_t) - J(\Phi_t^{\mathrm{ex}})}^2
            + 2 \langle J(\Phi_t) - J(\widetilde{\Phi}_t), J(\widetilde{\Phi}_t) - J(\Phi_t^{\mathrm{ex}}) \rangle
            \geq \norm*{J(\widetilde{\Phi}_t) - J(\Phi_t^{\mathrm{ex}})}^2.
        \end{multline}
    \end{widetext}
    The statement of the theorem is obtained by taking the square root on both sides on this inequality.
\end{proof}

From a practical point of view, the construction of the regularized dynamics $\{\widetilde{\Phi}_t\}_{t \geq 0}$ from a given dynamical map $\{\Phi_t\}_{t \geq 0}$ can now be summarized as follows.
At every time step, the Choi operator $J(\Phi_t)$ is obtained through evolution of the basis $\{E_{nm}\}$.
After diagonalization, one can check if $J(\Phi_t)$ is physical or not, and in the latter case the projection~\eqref{eq:projection-choi} is performed.
The result is $J(\widetilde{\Phi}_t)$ at all time steps.
At this point the regularization is concluded: the evolution of an arbitrary initial state $\rho$ can be obtained from $J(\widetilde{\Phi}_t)$ using Eq.~\eqref{eq:inverse_choi}.
Note that this last step is very cheap compared to the regularization task and can be performed repeatedly for any desired initial state.


\section{Examples} \label{sec:examples}

In order to show the advantageous effects implied by the Choi-proximity regularization, we test the method in two simple scenarios: a qubit amplitude damping model and a spin-boson model.
For these systems exact solutions are available and can be used as benchmarks.

\subsection{Qubit amplitude damping}\label{subsec:qubit}

Consider the case in which the system $S$ is a qubit with energy levels $\ket{0}$, $\ket{1}$, and Hamiltonian $H_S = \omega \dyad{1}$, where $\omega > 0$.
The environment is a bath of harmonic oscillators, so that $H_B = \sum_k \epsilon_k b_k^\dagger b_k$, where $b_k, b_k^\dagger$ are bosonic creation and annihilation operators.
The interaction is assumed to be a rotating-wave bilinear
\begin{equation}
    H_I = \dyad{1}{0} \otimes B + \dyad{0}{1} \otimes B^\dagger,
    \quad
    B = \sum_k g_k b_k,
\end{equation}
where $g_k$ are some constants.
It is easy to see that the environment is described by a single correlation function, which we assume for simplicity to be exponentially decaying in time:
\begin{equation}\label{eq:corr_function}
    c(t) = \frac{\gamma\mu}{2} e^{-\mu |t|} e^{-i\omega t}.
\end{equation}
In frequency space, this is equivalent to a Lorentzian spectral density in resonance with the qubit energy $\omega$.
The inverse of the fastest evolution time of the system can be estimated as~\cite{Mozgunov2020}
\begin{equation}
    \frac{1}{\tau_{SE}} \coloneqq \int_0^\infty |c(t)|dt = \frac{\gamma}{2},
\end{equation}
while the environment correlation time is
\begin{equation}
    \tau_E \coloneqq \frac{\int_0^\infty t|c(t)|dt}{\int_0^\infty |c(t)|dt} = \frac{1}{\mu}.
\end{equation}
One can then prove~\cite{Breuer2002} that the accuracy of the perturbative master equations above is controlled by the dimensionless ratio
\begin{equation}
    r \coloneqq \frac{\tau_E}{\tau_{SE}} = \frac{\gamma}{2\mu},
\end{equation}
which is required to be sufficiently smaller than one for the Born/Redfield equations to be accurate.

\begin{figure*}
    \centering
    \includegraphics[width=\linewidth]{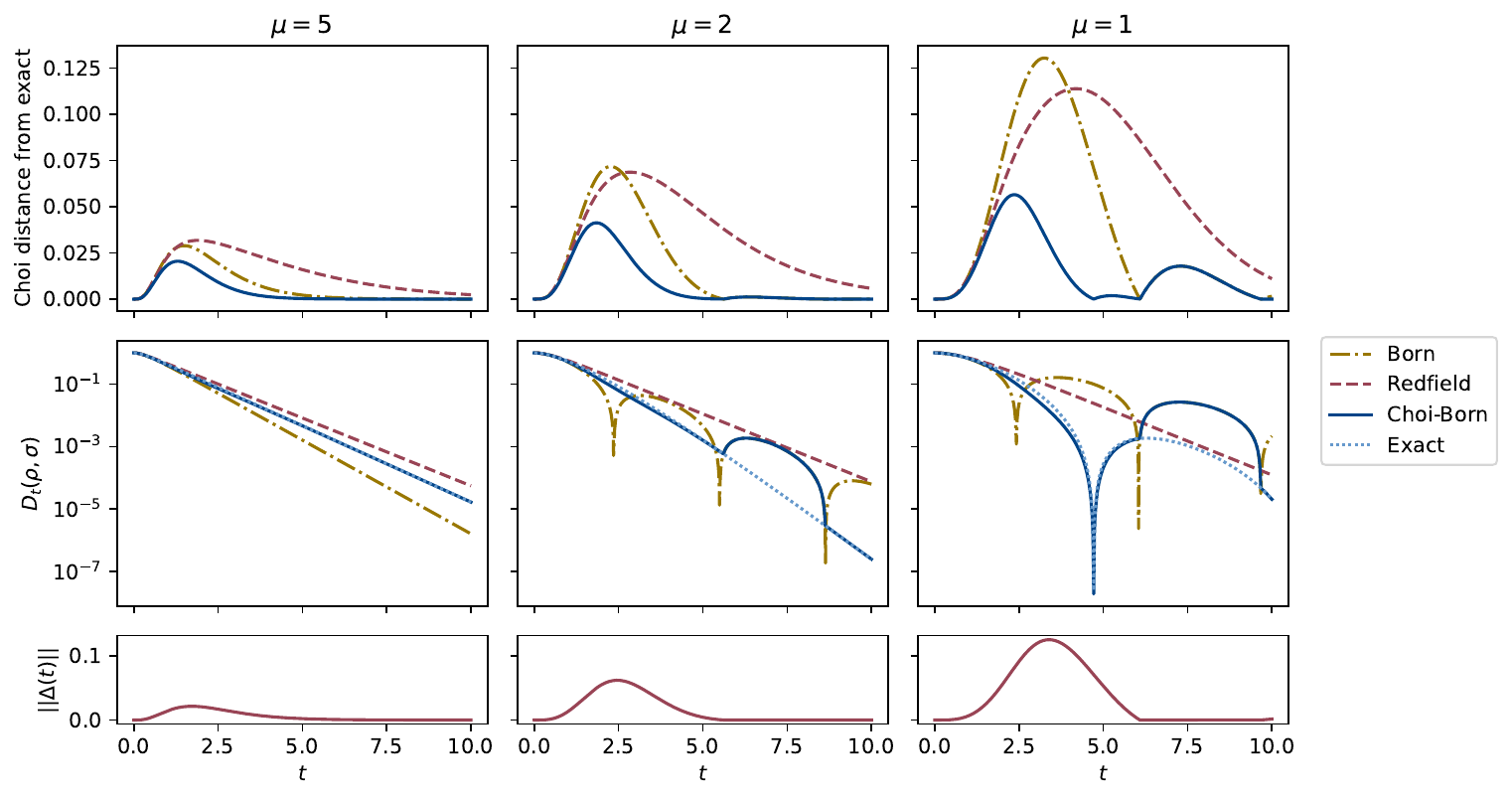}
    \caption{
        Plots obtained for the qubit system in Sec.~\ref{subsec:qubit} with $\gamma = 1$ and $\mu \in \{5, 2, 1\}$.
        Top row: dynamics of the Choi distance between a certain evolution and the exact one [cf.~Eq.~\eqref{eq:choi-distance}].
        Middle row: dynamics of the distinguishability measure $D_t(\rho, \sigma)$ in Eq.~\eqref{eq:distinguishability} with $\rho = \dyad{0}$ and $\sigma = \dyad{1}$.
        In this case the $y$-axis is plotted in logarithmic scale for better readability.
        Bottom row: CPT violation parameter $\norm{\Delta(t)}$, with $\Delta(t)$ defined in Eq.~\eqref{eq:cpt-violation}.
        Note how the Choi-proximity regularization of the Born master equation (``Choi-Born'' in the legend) outperforms the other methods and, when is activated, reproduces the Markovian and non-Markovian behaviors of the exact dynamics.
    }
    \label{fig:qubit}
\end{figure*}

This is probably the simplest nontrivial open quantum setting that can be solved exactly.
It is also possible to analytically solve both the Born equation and the Redfield equation~\cite{Breuer2002}.
All these solutions can be cast in the following form:
\begin{subequations}\label{eq:qubit_solutions}
    \begin{align}
        \rho_{00}(t) &= \rho_{00}(0) A(t) + 1 - A(t), \\
        \rho_{01}(t) &= \rho_{01}(0) B(t) e^{i\omega t}, \\
        \rho_{11}(t) &= \rho_{11}(0) A(t),
    \end{align}
\end{subequations}
where $\rho_{ij}(t) = \mel{i}{\rho(t)}{j}$ for $i,j \in \{0,1\}$.
Now, define the following function:
\begin{equation}
    G(\alpha,t) \coloneqq e^{-\mu t/2} \qty[\cosh(\frac{\alpha t}{2}) + \frac{\mu}{\alpha} \sinh(\frac{\alpha t}{2})].
\end{equation}
Then, the exact solution is characterized by
\begin{subequations} \label{eq:qubit-exact}
    \begin{align}
        A(t) &= |G(\alpha,t)|^2,
        \quad
        B(t) = G(\alpha,t), \\
        \alpha &= \sqrt{\mu^2 - 2\gamma\mu}.
    \end{align}
\end{subequations}
Instead, for the Born equation~\eqref{eq:born} one obtains
\begin{subequations} \label{eq:qubit-born}
    \begin{align}
        A(t) &= G(\alpha',t),
        \quad
        B(t) = G(\alpha,t), \\
        \alpha' &= \sqrt{\mu^2 - 4\gamma\mu},
        \quad
        \alpha = \sqrt{\mu^2 - 2\gamma\mu}.
    \end{align}
\end{subequations}
Finally, for the time-dependent version of the Redfield equation in Eq.~\eqref{eq:redfield}, one finds
\begin{subequations} \label{eq:qubit-redfield}
    \begin{align}
        A(t) &= e^{-R(t)},
        \quad
        B(t) = e^{-R(t)/2}, \\
        R(t) &= \gamma\qty(t + \frac{e^{-\mu t} - 1}{\mu}).
    \end{align}
\end{subequations}
The time-independent variant is instead characterized by $R(t) = \gamma t$.
However, for the purpose of the present example we will focus exclusively on the more complex time-dependent scenario.

First, let us ask about the Markovianity of these solutions.
Given two arbitrary qubit states $\rho, \sigma$ it is easy to calculate that, under the evolution dictated by Eq.~\eqref{eq:qubit_solutions}, their distinguishability is
\begin{equation}
    D_t(\rho, \sigma) = \sqrt{|\rho_{11} - \sigma_{11}|^2 A^2(t) + |\rho_{01} - \sigma_{01}|^2 B^2(t)}.
\end{equation}
In case $B^2(t) = A(t)$ for all $t \geq 0$ (which happens for the exact solution and the Redfield equation), Eq.~\eqref{eq:qubit_solutions} describes a qubit amplitude damping channel, and it is known that in this case the BLP condition~\eqref{eq:BLP} provides a necessary and sufficient condition for CP-divisibility~\cite{Lonigro2022,Chruscinksi2022}.
In this case,
\begin{equation}\label{eq:trace_dist_A}
    D_t(\rho, \sigma) = |\rho_{11} - \sigma_{11}| A(t),
\end{equation}
and CP-divisibility is equivalent to $dA(t)/dt \leq 0$ for all $t \geq 0$.
Notice that
\begin{equation} \label{eq:blp_redfield}
    \frac{d}{dt} e^{-R(t)} = -\gamma e^{-R(t)} (1-e^{-\mu t}) \leq 0,
\end{equation}
therefore the Redfield solution is always Markovian.
Instead, for the exact solution we have
\begin{equation}
    \frac{d}{dt} |G(\alpha,t)|^2 = \qty(1-\frac{\mu^2}{\alpha^2}) \alpha e^{-\mu t/2} \sinh(\frac{\alpha t}{2}) G(\alpha, t).
\end{equation}
In case $\mu \geq 2\gamma$ (i.e., $r \leq 1/4$), we have that $\alpha$ is a nonnegative real number that satisfy $\mu^2/\alpha^2 \geq 1$, and $G(\alpha,t) \geq 0$.
This implies that the above derivative is nonpositive and the dynamics is Markovian.
On the other hand, if $\mu < 2\gamma$ (i.e., $r > 1/4$) then $\alpha$ is an imaginary number: the above derivative oscillates, and can take positive values.
This means that the dynamics is non-Markovian for sufficiently small spectral width.

In case $B^2(t) \neq A(t)$ (which happens for the Born equation), the BLP condition is only necessary for Markovianity.
Nevertheless, for initial states $\rho, \sigma$ with equal coherences, the distinguishability function $D_t(\rho, \sigma)$ still acquires the form in Eq.~\eqref{eq:trace_dist_A}, so that we can again focus our attention on $A(t)$ only.
Now,
\begin{equation}
    \frac{d}{dt} G(\alpha', t) = \frac{1}{2} \qty(1-\frac{\mu^2}{\alpha'^2}) e^{-\mu t/2} \sinh(\frac{\alpha' t}{2}).
\end{equation}
Again, if $\mu < 4\gamma$ (i.e., $r > 1/8$) this derivative oscillates and can acquire positive values: this indicates that for sufficiently small spectral width the Born equation is non-Markovian.

Obviously, we expect the Born and Redfield equations to be inaccurate for small $\mu$, but already at moderate couplings (with $r \sim 1/4$) the Redfield equation is completely unable to capture the true non-Markovian character of the dynamics.
The Born equation, instead, is able to show this behavior.

Now let us ask about the CPT requirement.
A straightforward computation shows that the Choi operator associated with the dynamical map described by Eq.~\eqref{eq:qubit_solutions} is given by
\begin{equation}\label{eq:qubit_choi}
    P(t) = \frac{1}{2} \begin{bmatrix}
        1 & 0 & 0 & B(t)e^{i\omega t} \\
        0 & 1-A(t) & 0 & 0 \\
        0 & 0 & 0 & 0 \\
        B(t)e^{-i\omega t} & 0 & 0 & A(t)
    \end{bmatrix}.
\end{equation}
This always satisfies $\Tr_1 P(t) = \mathbbm{1}/2$, in accordance with the fact that all the involved dynamical maps are trace-preserving.
For what concerns its positivity, a straightforward computation reveals that the eigenvalues of $P(t)$ are
\begin{equation}
    0, \,\, \frac{1-A(t)}{2}, \,\,
    \frac{1 + A(t) \pm \sqrt{[1-A(t)]^2 + 4B^2(t)}}{4}.
\end{equation}
Since $A(t) \leq 1$, these eigenvalues are always nonnegative except (at most) the rightmost one with the minus sign, which is easily seen to be nonnegative if and only if $B^2(t) \leq A(t)$.
This condition is always satisfied by the exact solution~\eqref{eq:qubit-exact} and the Redfield equation~\eqref{eq:qubit-redfield}, but not by the Born equation~\eqref{eq:qubit-born}, which can then violate the complete positivity requirement.

We can now summarize the situation as follows.
The Redfield equation is CPT but it doesn't capture the non-Markovian features at moderate couplings.
On the other hand, the Born equation exhibits non-Markovian effects but in general it is not CPT.
It is then reasonable to apply the regularization procedure described in Sec.~\ref{sec:regularization} to the Born equation.

In order to compare the performances of the various master equations, we need an accuracy measure (with respect to the exact solution) that is independent from the initial state of the dynamics.
Here we choose the quantity considered in Ref.~\cite{DAbbruzzo2023}, referred to as ``Choi distance'' between an approximated dynamical map $\Phi_t$ and the exact one $\Phi_t^{\mathrm{ex}}$:
\begin{equation} \label{eq:choi-distance}
    \norm*{J(\Phi_t) - J(\Phi_t^{\mathrm{ex}})},
\end{equation}
where the Choi operators are calculated using Eq.~\eqref{eq:qubit_choi}.
If $\Phi_t$ is not CPT, by Theorem~\ref{thm:distance-from-exact} we know that this distance can always be lowered by taking the Choi-proximity regularization of $\Phi_t$.

In Fig.~\ref{fig:qubit} (top row) we plot these quantities for the Born equation, the Redfield equation, and the regularized version of the Born equation.
Here we fixed $\gamma = 1$ and we studied what happens for several values of $\mu$.
Obviously, a greater $\mu$ (and hence a smaller $r$) leads to better accuracy for all master equations, while things get worse when we approach the non-Markovian regime for smaller values of $\mu$.
However, notice that in all cases the regularized Born equation outperforms the other two: for the Born equation this is expected, but it was not obvious for the Redfield case.
Notice also that for small enough $\mu$ the Born equation develops regions in which it is CPT: unfortunately, our regularization cannot do anything there.

In Fig.~\ref{fig:qubit} (bottom row) we also show the dynamics of the distinguishability measure $D_t(\rho, \sigma)$ in Eq.~\eqref{eq:distinguishability} when $\rho = \dyad{0}$ and $\sigma = \dyad{1}$.
As before, $\gamma = 1$ and several values of $\mu$ are investigated.
In the Markovian regime (great values of $\mu$) we see how the prediction of the regularized Born equation is practically exact, while at this level some deviations can be seen for the other two master equations.
When entering the non-Markovian regime with smaller values of $\mu$, several things can be noticed.
The prediction of the Redfield equation is always monotonic, as expected from Eq.~\eqref{eq:blp_redfield}.
The Born equation, instead, is capable of showing non-monotonic behavior but it is quite different from the exact one, as a consequence of the violation of $r \ll 1$.
The regularized Born equation not only shows non-monotonic behavior as well, but also remarkably follows the exact solution in those time regions where the Born equation is not CPT.
Again, the regularization has no power in those regions where the Born equation is already CPT.
Finally, notice that there is an intermediate situation ($\mu = 2$) where the Born equation predicts non-Markovianity when the exact solution does not: when activated, our regularization is capable of correcting this behavior.


\subsection{Spin-boson model}\label{subsec:spinboson}

As a second example we study the more complex spin-boson model, which occupies a prominent role in the theory of open quantum systems and constitutes a paradigmatic playground to study many dissipative and decoherence mechanisms~\cite{Leggett1987,Breuer2002}.
The system consists of a spin-1/2 particle with Hamiltonian
\begin{equation}
    H_S = \frac{\varepsilon}{2} \sigma_z + \frac{\delta}{2} \sigma_x,
\end{equation}
where $\sigma_x, \sigma_y, \sigma_z$ are the usual Pauli matrices, $\varepsilon$ is the energy gap between ground state $\ket{0}$ and excited state $\ket{1}$, and $\delta$ is a tunneling strength between such states.
The spin is coupled to an environment of bosonic oscillators $H_E = \sum_k \epsilon_k b_k^\dagger b_k$, as in Sec.~\ref{subsec:qubit}, and the interaction is
\begin{equation}
    H_I = \sigma_z \otimes \sum_k g_k \qty(b_k + b_k^\dagger).
\end{equation}
In order to keep things simple, we assume again that the environment is characterized by an exponential correlation function [cf.~Eq.~\eqref{eq:corr_function}]
\begin{equation}
    c(t) = \frac{\gamma\mu}{2} e^{-\mu|t|} e^{-i\omega_0 t},
\end{equation}
where $\mu$ is the spectral width, $\gamma$ is the coupling constant, and $\omega_0$ is the resonance frequency.

There is no explicit exact solution for this model, even though the system is small enough to allow for an easy numerical approach through hierarchical equations of motion (HEOM)~\cite{Tanimura2020}.
In particular, we rely on the corresponding functions provided by the Python library QuTip~\cite{Johansson2013}: in the following, the result obtained with this procedure will be referred to as ``exact solution''.
The Redfield equation associated with this problem (and its regularizations) can also easily be solved using the vectorization approach illustrated in App.~\ref{app:vectorization}.

It is generally believed that CPT violations of the Redfield equation tends to appear at short times and with at least moderate coupling constant~\cite{Hartmann2020} (even though some positivity-violation effects have been detected at long times too~\cite{Benatti2022}, and the onset of positivity violations is not a mathematical guarantee of complete breakdown of the approximations used to derive the master equation).
In this regime the accuracy with respect to the exact solution is not expected to be satisfying.
Applying the Choi-proximity regularization can slightly improve performances, but our aim in this example is another one: we want to show how our procedure can give a well-defined dynamics that retains the non-Markovian features of the Redfield equation, contrary to other common regularization schemes.
For completeness, in App.~\ref{app:inaccuracy} we show that in the regime studied below (where positivity violation occurs) the Redfield equation is indeed inaccurate for the spin-boson model, in accordance with what is discussed in Ref.~\cite{Hartmann2020}.

\begin{figure}
    \centering
    \includegraphics[width=\columnwidth]{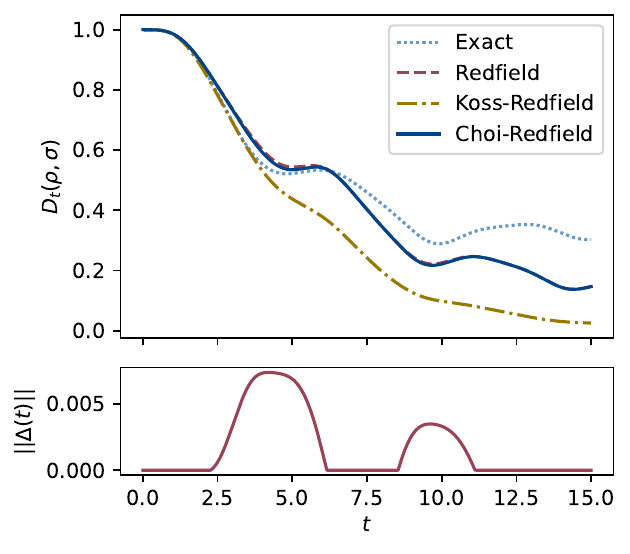}
    \caption{
        Dynamics of the distinguishability measure $D_t(\rho, \sigma)$ of the spin-boson model, with $\rho = \dyad{0}$ and $\sigma = \dyad{1}$.
        At the bottom, the CPT violation parameter $\norm{\Delta(t)}$ [cf.~Eq.~\eqref{eq:cpt-violation}].
        The parameters of the model are $\varepsilon = 1, \delta = 0.7$ for the system and $\gamma = 1.5, \mu = 0.1, \omega_0 = 1$ for the interaction.
    }
    \label{fig:spin_boson}
\end{figure}
\begin{figure}
    \centering
    \includegraphics[width=\columnwidth]{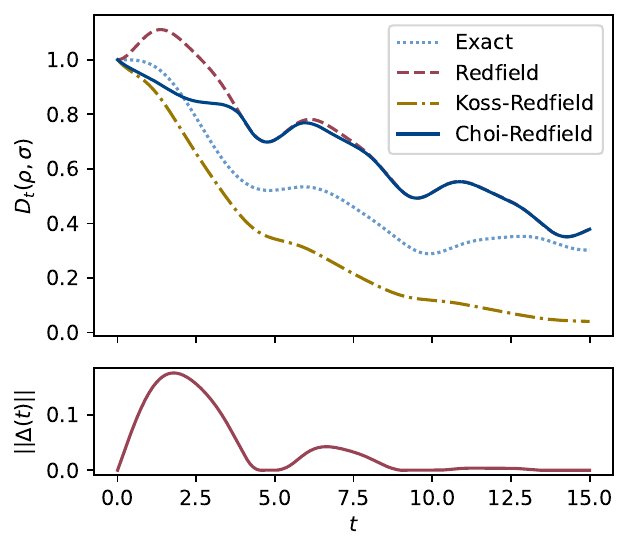}
    \caption{
        Same as Fig.~\ref{fig:spin_boson} but using a Redfield equation with time-independent generator (see discussion below Eq.~\eqref{eq:redfield}).
    }
    \label{fig:spin_boson_ind}
\end{figure}
Let us first focus on the time-dependent Redfield equation.
In Fig.~\ref{fig:spin_boson} we show the dynamics of the distinguishability measure $D_t(\rho, \sigma)$ when $\rho = \dyad{0}$ and $\sigma = \dyad{1}$ in a moderate coupling regime with $\gamma = 1.5$ and $\mu = 0.1$, together with the CPT violation parameter $\norm{\Delta(t)}$.
In the distinguishability plot we see that the regularization does not make much of a difference with respect to the original time-dependent Redfield evolution, but the dynamics is now guaranteed to be physically well defined.
Notice how the time-dependent Redfield equation violates positivity at short times with $\norm{\Delta(t)} \sim 10^{-3}$.
The moderate impact of the Choi-proximity regularization allows the regularized equation to keep the nice features of the Redfield equation, such as accounting for the non-Markovianity of the model: note how the oscillations in the distinguishability measure are roughly positioned as in the exact solution.
In contrast, we also plot the prediction of the time-dependent Redfield equation regularized by acting on its Kossakowski matrix~\cite{DAbbruzzo2023}, as described at the end of Sec.~\ref{subsec:master-equations}: in this case a monotonic decrease is observed, in accordance with the fact that we are imposing a CP-divisibility requirement.

In Fig.~\ref{fig:spin_boson_ind} we show the same plot but starting from the time-independent version of the Redfield equation.
This time the positivity violation is clearly visible also in the distinguishability plot, where $D_t(\rho,\sigma)$ unphysically grows beyond the value of one at short times.
Correspondingly, we find $\norm{\Delta(t)} \sim 10^{-1}$ in this time region.
The Choi-proximity regularization corrects this behavior and reproduces non-Markovian oscillations, contrary to the Kossakowski-regularized equation.


\section{Conclusions} \label{sec:conclusions}

In this work we proposed a numerical approach that takes any Hermitian-preserving quantum evolution and transforms it into the closest trace-preserving and completely positive one, based on a projection operation in the Choi space.
By construction, it is sufficient to perform this regularization only once to be able to evolve an arbitrary initial state.
Moreover, given the strong link with quantum process tomography tasks, any future improvement in computational complexity in that field is likely to automatically provide an improvement for our regularization too.
We put forward evidence that this procedure improves the accuracy of the original map, avoids unphysical predictions, and retains eventual non-Markovian features.
We also provided the formal expressions of paradigmatic master equations that result from the Choi-proximity regularization procedure.

Here for simplicity we only showed examples using Redfield and Born equations, but this method can also be used for higher-order master equations, and in general whenever there is reason to believe that complete positivity or trace preservation could be violated in some time region of the evolution at hand.
The method could also be applied outside the standard open-quantum-systems setting, for example in the description of driven quantum systems or disorder-averaged quantum evolutions, where Redfield-like equations are known to occur: see, e.g., Refs.~\cite{Gneiting2020,Groszkowski2023} and references therein.

At this level, the procedure is intended to be mainly practical and it is expected to be useful mainly to avoid false predictions of positivity violation.
No clear physical insight was provided about the regularized dynamics, and it is still an open problem to understand what are the implications of the Choi-proximity regularization, e.g., on the thermodynamics of the system and the steady-state manifold structure.
Our intuition is that CPT restoration puts effectively back some non-Markovian effects that are neglected during the derivation of the master equations under analysis, an interpretation sustained by the apparence of an explicit dependence on the initial state in the modified master equations.


\begin{acknowledgments}
    The authors acknowledge useful discussions with Vasco Cavina.
    V.G. and A.D. acknowledge financial support by MIUR (Ministero dell'Istruzione, dell'Università e della Ricerca) via PRIN 2017 ``Taming complexity via QUantum Strategies: a Hybrid Integrated Photonic approach'' (QUSHIP), Id. 2017SRN-BRK, via project PRO3 ``Quantum Pathfinder'', via the PNRR MUR project PE0000023-NQSTI, and via PRIN 2022 ``Recovering Information in Sloppy QUantum modEls'' (RISQUE), Id. 2022T25TR3.
    D.F. acknowledges support from the Government of Spain (Severo Ochoa CEX2019-000910-S and TRANQI), Fundació Cellex, Fundació Mir-Puig, Generalitat de Catalunya (CERCA program) and ERC AdG CERQUTE.
    D.F. acknowledges financial support from PNRR MUR Project No. PE0000023-NQSTI.
\end{acknowledgments}


\appendix

\section{Algorithm for Choi-space projection} \label{app:algorithm}

In this appendix we discuss how to algorithmically solve the semidefinite least-squares problem~\eqref{eq:projection-choi} of projecting a Hermitian operator in $\mathrm{Herm}(\mathcal{H} \otimes \mathcal{H})$ onto the physical Choi space $\mathfrak{J}(\mathcal{H})$ (see Sec.~\ref{subsec:projection}).

A natural approach to tackle these kinds of problems is casting them as standard semidefinite programs, for which general purpose algorithm exist~\cite{Nocedal2006}.
However, if one recognizes that $\mathfrak{J}(\mathcal{H})$ is the intersection of two convex sets, tailored approaches can be devised~\cite{Henrion2011}.
In the context of quantum process tomography, this idea was explored and improved in Ref.~\cite{Stepney2022} for the specific problem in Eq.~\eqref{eq:projection-choi}, where a so-called hyperplane intersection projection algorithm was devised.

Here we arbitrarily choose to use an alternative approach based on Lagrangian duality~\cite{Urruty1993}.
Specifically, we make use of the following result, which is an adaptation of what was originally proved in Ref.~\cite{Malick2004} in the context of generic semidefinite least-squares problems.
\begin{theorem}\label{thm:dual}
    Problem~\eqref{eq:projection-choi} is solved by
    \begin{equation}
        \widetilde{P} = \Pi(P + \mathbbm{1}_\mathcal{H} \otimes \overline{Y}),
    \end{equation}
    where
    \begin{equation}
        \Pi: A \to \frac{A + \sqrt{A^\dagger A}}{2}
    \end{equation}
    is the projector onto $\mathrm{Pos}(\mathcal{H} \otimes \mathcal{H})$ and $\overline{Y}$ is the solution of the following unconstrained optimization problem:
    \begin{equation}\label{eq:Y_problem}
        \overline{Y} = \operatorname*{arg\,min}_{Y \in \mathrm{Herm}(\mathcal{H})} \theta(Y),
    \end{equation}
    where
    \begin{equation}
        \theta(Y) = \frac{1}{2} \norm{\Pi(P + \mathbbm{1}_\mathcal{H} \otimes Y)}^2 - \Tr Y
    \end{equation}
    is a convex differentiable function whose gradient with respect to $Y$ is the following Lipschitz-continuous function:
    \begin{equation}
        \nabla \theta(Y) = \Tr_1[\Pi(P + \mathbbm{1}_\mathcal{H} \otimes Y)] - \frac{\mathbbm{1}_\mathcal{H}}{d}.
    \end{equation}
\end{theorem}

Notice that the action of $\Pi$ on a Hermitian operator $A$ can easily be evaluated once we have a spectral decomposition of $A$: simply put to zero all the negative eigenvalues while leaving eigenvectors alone.
More explicitly,
\begin{equation}
    \Pi(U \mathrm{diag}(\bm{\lambda}) U^\dagger) = U \mathrm{diag}(\max\{\bm{\lambda}, \bm{0}\}) U^\dagger,
\end{equation}
where $\bm{\lambda}$ is the vector of eigenvalues, $\mathrm{diag}(\mathbf{x})$ is the diagonal matrix having the entries of the vector $\mathbf{x}$ on the diagonal, and $\max\{\bm{\lambda}, \bm{0}\}$ is the vector having as entries $\max\{\lambda_i, 0\}$ for $\lambda_i$ component of $\bm{\lambda}$.

The advantage offered by Theorem~\ref{thm:dual} is that now we can focus on solving the unconstrained problem~\eqref{eq:Y_problem} using nonlinear convex optimizer, such as accelerated gradient descent or quasi-Newton techniques~\cite{Nocedal2006}.
Here we choose the Broyden-Fletcher-Goldfarb-Shanno (BFGS) algorithm provided by the routine \texttt{optimize.minimize} of the Python library SciPy~\cite{Virtanen2020}.
Moreover, we only need to track $d^2$ real parameters in the $Y$ space instead of $d^4$ in $X$ space.
This can be particularly relevant when working with high-dimensional systems.
Note, however, that the computational cost of the procedure is still dominated by the $\Pi$ operation, which requires a diagonalization of a $d^2 \times d^2$ matrix and approximately scales as $\mathcal{O}(d^6)$~\cite{Wendland2017}.

Up to now, it is unclear how this dual approach relates to the more recent hyperplane intersection projection algorithm, even though preliminary results can already be found in Ref.~\cite{Stepney2022}, where it is shown that the dual approach is still competitive as a state-of-art method for solving~\eqref{eq:projection-choi}.


\section{Vectorization of the Redfield equation}\label{app:vectorization}

In this appendix we discuss how to vectorize the Kossakowski form of the Redfield equation~\eqref{eq:redfield-standard} in order to allow for an easy numerical evaluation of its solution, while keeping at every stage explicit access to the Kossakowski matrix $\chi(t)$.
Start by noticing that $E^\dagger_{nm} E_{kq} = \delta_{kn} E_{mq}$, and therefore we have $H_S + H_{LS}(t) = \sum_{k,q} h_{kq}(t) E_{kq}$ with
\begin{equation}
    h_{kq}(t) = \delta_{kq}\omega_k + \sum_\ell \eta_{\ell q, \ell k}(t),
\end{equation}
where we also used the fact that $H_S = \sum_k \omega_k E_{kk}$.
With the same reasoning, we can write
\begin{equation}
    \sum_{k,q,n,m} \chi_{kq,nm}(t) \{E_{nm}^\dagger E_{kq}, \rho\} = \sum_{k,q} \widetilde{\chi}_{kq}(t) \{E_{kq}, \rho\},
\end{equation}
where
\begin{equation}
    \widetilde{\chi}_{kq}(t) = \sum_\ell \chi_{\ell q, \ell k}.
\end{equation}
We arrive at
\begin{multline}\label{eq:koss_modified}
    \frac{d}{dt} \rho(t) = \sum_{k,q,n,m} \chi_{kq,nm}(t) E_{kq} \rho(t) E_{nm}^\dagger \\
    - \sum_{k,q} \qty(i h_{kq}(t) [E_{kq}, \rho(t)] + \frac{1}{2} \widetilde{\chi}_{kq}(t) \{E_{kq}, \rho(t)\}).
\end{multline}
After expanding commutator and anticommutator, this can be rewritten as
\begin{multline}
    \frac{d}{dt} \rho(t) = \sum_{k,q,n,m} \chi_{kq,nm}(t) E_{kq} \rho(t) E_{nm}^\dagger \\
    + \sum_{k,q} \qty(\phi_{kq}(t) E_{kq} \rho(t) + \phi^*_{qk}(t) \rho(t) E_{kq}),
\end{multline}
where we introduced for convenience
\begin{equation}
    \phi_{kq}(t) \coloneqq -i h_{kq}(t) - \frac{1}{2} \widetilde{\chi}_{kq}(t).
\end{equation}

Now we introduce the row-major vectorization operator $\mathrm{vec}: \mathrm{L}(\mathcal{H}) \to \mathcal{H} \otimes \mathcal{H}$ defined by $\mathrm{vec}(e_n e_m^\dagger) = e_n \otimes e_m$ on any orthonormal basis $\{e_n\}$ of $\mathcal{H}$.
More explicitly, if $A \in \mathrm{L}(\mathcal{H})$ is given in matrix form, $\mathrm{vec}(A)$ is the vector obtained by transposing the rows of $A$ and stacking them on top of one another.
For any $A,B,X \in \mathrm{L}(\mathcal{H})$ it can be proved that
\begin{equation}
    \mathrm{vec}(AXB) = (A \otimes B^T) \mathrm{vec}(X),
\end{equation}
Using this relation together with $r(t) \coloneqq \mathrm{vec}(\rho(t))$, we obtain that Eq.~\eqref{eq:koss_modified} turns into the vectorized equation
\begin{equation}\label{eq:red_vectorized}
    \frac{d}{dt} r(t) = \mathcal{L}(t) r(t),
\end{equation}
where
\begin{equation}
    \mathcal{L}(t) = \mathcal{X}(t) + \phi(t) \otimes \mathbbm{1} + \mathbbm{1} \otimes \phi^*(t)
\end{equation}
and we introduced the matrix $\mathcal{X}(t)$ whose elements are deduced from the Kossakowski matrix as follows:
\begin{equation}
    \mel{k,n}{\mathcal{X}(t)}{q,m} = \chi_{kq,nm}(t),
\end{equation}
where $\ket{k,n} \equiv \ket{k} \otimes \ket{n}$.

Similar expressions were found in Ref.~\cite{DAbbruzzo2023} for a particular model, while these are generalized to any Redfield equation.
Eq.~\eqref{eq:red_vectorized} can be solved with any ordinary differential equation solver: for example, the results of this paper are obtained using the routine \texttt{integrate.solve\_ivp} in the Python library SciPy~\cite{Virtanen2020}, which uses by default an explicit Runge-Kutta method of order 5(4).
The state matrix $\rho(t)$ in the canonical basis $\{E_{kq}\}$ can then be found from $r(t)$ by inverting the vectorization mapping.


\section{Inaccuracy of the Redfield equation for the spin-boson model with moderate coupling} \label{app:inaccuracy}

Complementing the discussion of Sec.~\ref{subsec:spinboson}, in this appendix we show that the Redfield equation for the spin-boson model is inaccurate for moderate couplings (where positivity violation occurs), confirming the findings of Ref.~\cite{Hartmann2020}.
Specifically, we report in Fig.~\ref{fig:spin_boson_acc} the Choi distance~\eqref{eq:choi-distance} from the time-independent Redfield equation and the exact solution, using the same parameters used in Fig.~\ref{fig:spin_boson_ind}.
It is evident that this distance is sensibly different from zero in all time regions, being of order one.
We also show the performance after a Choi-proximity regularization is applied: an improvement is indeed found (as predicted by Theorem~\ref{thm:distance-from-exact}), but it is not dramatic. The same is true using the time-dependent version of the Redfield equation, as also shown in Fig.~\ref{fig:spin_boson_acc}.

\begin{figure}[b]
    \centering
    \includegraphics[width=\columnwidth]{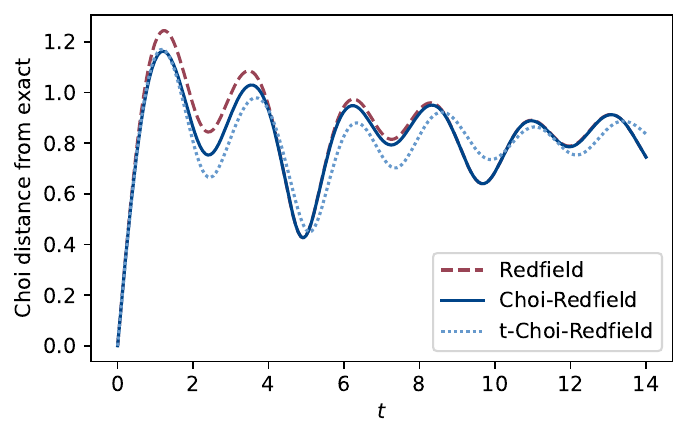}
    \caption{
        Choi distance between the exact solution of the spin-boson model in Sec.~\ref{subsec:spinboson} and various master equations: ``Redfield'' means the time-independent Redfield equation, ``Choi-Redfield'' means its Choi-proximity regularization, while ``t-Choi-Redfield'' stands for the regularized time-dependent Redfield equation.
        The parameters are the same as in Fig.~\ref{fig:spin_boson_ind}.
    }
    \label{fig:spin_boson_acc}
\end{figure}


\bibliography{bibliography.bib}

\end{document}